\documentclass[twoside]{article}
\usepackage{latexsym,amsmath,amssymb}
\usepackage{hyperref}
\usepackage{cite} 
\usepackage{sectsty}
  \sectionfont{\large} \subsectionfont{\normalsize}
\usepackage{graphicx}

\usepackage{theorem}
\usepackage[all]{xy}
\newtheorem{theorem}{Theorem}[section]
{\theorembodyfont\rmfamily
\newtheorem{lemma}{Lemma}[section]

\newtheorem{corollary}{Corollary}

}
\newenvironment{proof}{\begin{paragraph}
          {Proof}}{\end{paragraph}}

\renewenvironment{abstract}
 {\small\begin{quote}{\textbf{Abstract}}\,\,}{\end{quote}}
\newenvironment{keywords}
 {\small\begin{quote}{\textbf{Keywords}}\,\,}{\end{quote}}
\newenvironment{classification}
 {\small\begin{quote}{\textbf{2010 Mathematics Subject Classification}}\,\,}{\end{quote}}
\date{}

\title{\vspace{-9ex}
{\centering
 \textbf{\large Variational inequality for perpetual American option price and convergence to the solution of the difference equation
}}}

 \vspace{-4ex}

\author{\small\textsf{\bfseries 
$^{1}$ Hyong-chol O, $^{2}$ Song-San Jo $^{2}$}\\[-.5ex]
{\footnotesize  ${}^{1, 2}$ Faculty of Mathematics, \textbf{Kim Il Sung} University,}
{\footnotesize   Pyongyang , D P R Korea}\\[-.5ex]
{\footnotesize e-mail:  $^{1}$hc.o@ryongnamsan.edu.kp }}

\begin{document}

\maketitle
\thispagestyle{empty}

\vspace{-.6cm}

\begin{abstract}
A variational inequality for pricing the perpetual American option and the corresponding difference equation are considered. First, the maximum principle and uniqueness of the solution to variational inequality for pricing the perpetual American option are proved. Then the maximum principle, the existence and uniqueness of the solution to the difference equation corresponding to the variational inequality for pricing the perpetual American option and the solution representation are provided and the fact that the solution to the difference equation converges to the viscosity solution to the variational inequality is proved. It is shown that the limits of the prices of variational inequality and BTM models for American Option when the maturity goes to infinity do not depend on time and they become the prices of the perpetual American option.
\end{abstract}



%
%

\section{Introduction}

\indent
Pricing financial derivatives is one of main topics in mathematical finance with clear implications in physics \cite{Miq1} and the American option is one of the widely studied financial derivatives.

\cite{Jia} studied discrete and differential equation models for pricing options and provided various pricing formulae. In particular, they provided the free boundary problem and variational inequality models for the prices of American options and studied some properties of the option prices. American options are one of financial instruments that the holder may exercise at expiry dates and any time before expiry dates. 

Perpetual American options are the American options without expiry dates, that is, perpetual American options can be exercised at any time in the future and the price functions are defined on infinite time intervals. A pricing model for the perpetual American options was studied in \cite{GS} using expectation method. Under the diffusion model, \cite{Jia} provided free boundary problem and variational inequality models for perpetual American options and the solution representation for the free boundary problem model. Under a general jump-diffusion risk model, \cite{Chi} studied perpetual American options. In \cite{JP},\cite{Miq1},\cite{Miq2} perpetual American options are studied within a market model that can be considered as a regime-switching models with only two states and where one of the transition probabilities is set to zero. In \cite{SW} perpetual American catastrophe equity put options in a jump-diffusion model are studied and \cite{SY} studied perpetual American executive stock options (ESO), \cite{AB1},\cite{AB2} studied call-put dualities when the local volatility function is modified.

On the other hand, financial contracts are discretely exercised, for example, every day, every month, every three months, every six months, every year and etc. So the discrete models are considered as more realistic models for financial derivatives than continuous models \cite{Jia}.

For perpetual American options, \cite{LL} provided the price representation and optimal exercise boundary by the binomial tree model which is one of most popular discrete models, and they studied the periodicity of the binomial tree price of perpetual Bermudan options. \cite{JD} found some clear relationship between the binomial tree method and a special explicit difference schemes for the variational inequality models for American option’s price and using it obtained the convergence of the binomial tree price for American option to the viscosity solution to the variational inequality model. This result is extended to the case with time dependent coefficients in \cite{OJK}.

In this article we are interested in the convergence of the binomial tree price for perpetual American option obtained in \cite{LL}. The price functions for American options are defined on finite time intervals whereas the price functions for perpetual American options are defined on infinite time intervals. So the theory developed in \cite{JD} or \cite{OJK} could not cover this problem and the solutions to discrete models such as the binomial tree model cannot be easily found step by step in difference schemes as shown in \cite{LL}.

In this article we obtain the maximum principle and uniqueness of the solution to variational inequality for pricing the perpetual American options. Then we study the maximum principle, the existence and uniqueness of the solution to the difference equation corresponding to the variational inequality for pricing the perpetual American option and the solution representation. (From the consideration in \cite{JD},\cite{Jia}, the binomial tree method of \cite{LL} can be considered as a special difference equation for the variational inequality neglecting infinitesimal and thus such study can be viewed as an extension of the results of \cite{LL}.) Then we study the convergence of the solution to the difference equation to the viscosity solution to the variational inequality. From this result, we have the convergence of the binomial tree method of \cite{LL} and such study can be viewed as an extension of the results of \cite{JD} to the case with the expiry time $T=\infty$.

In many literature including \cite{Jia} and \cite{LL}, the price of perpetual American option is modeled under the apriori assumtion that it does not depend on time. On the other hand, it is natural to consider  the price of perpetual American option as a limit of the price of American Option when the maturity goes to infinity. This approach excludes the apriori assumtion that the price of perpetual American option does not depend on time. In this paper we show that the limits of the prices of variational inequality and BTM models for American Option when the maturity goes to infinity do not depend on time.

The rest of this article is organized as follows. In Section 2 the maximum principle and the uniqueness of the solution to variational inequality for pricing the perpetual American option, the existence of the optimal exercise boundary are described. And then For the difference equation corresponding to the variational inequality for pricing the perpetual American option, the maximum principle, the existence of optimal exercise boundary, the existence, uniqueness and representation of the solution and the convergence of the approximated solution are discussed. Section 3 shows that the limits of the prices of variational inequality model, its explicit difference scheme and BTM for American Option when the maturity goes to infinity do not depend on time.

%
%

\section{The variational inequality for perpetual American option price and its difference equation}
\indent

\subsection{Some properties of the solution to the variational inequality for perpetual American option price}
Let $r$, $q$ and $\sigma$ be the interest rate, the dividend rate and the volatility of the underlying asset of option, respectively, then a variational inequality pricing model of American option is provided as follows:
\begin{equation}\label{eq 1}
\min\left\{-\frac{\sigma^2}{2}S^2\frac{d^2V}{dS^2}-(r-q)S\frac{dV}{dS}+rV,V-\psi\right\}=0,
\end{equation}
Here $\psi=(E-S)^+(for~put)$~~or~~$\psi=(S-E)^+(for~call).$

The Black-Scholes ordinary differential operator (BSOD operator) is defined as follows:

\begin{equation}\label{eq 2}
\-LV=-\frac{\sigma^2}{2}S^2\frac{d^2V}{dS^2}-(r-q)S\frac{dV}{dS}+rV
\end{equation}

If $V(S)$ is the solution of $\eqref{eq 1}$, then we have $V(S)=\psi,~-LV>0$ in stopping region, and $V(S)>\psi,~-LV=0$ in continuation region.[3]
Thus, the solution of \eqref{eq 1} is always nonnegative. Consider the BSOD operator $\eqref{eq 2}$ in the interval $A=(a,b),~(0\le{a<b}\le\infty).$

\begin{theorem}\label{theorem 1}
(Maximum principle of BSOD operator)
Suppose that $r>0$ and $V(S)\in{C^2(A)}$. If $-LV<(>)0,~S\in{A}$, then nonnegative maximum (nonpositive minimum) value of $V$ cannot be attained at the interior points of $A$. Moreover, if $-LV\le(\ge)0,~S\in{A}$, then we have
\begin{equation}\label{eq 3}
\sup_{x\in{A}}V(x)=\sup_{x\in{\partial{}A}}V^+(x)~~\left(\inf_{x\in{A}}V(x)=\inf_{x\in\partial{A}}V^-(x)\right).
\end{equation}
\end{theorem}

\begin{proof} 
Suppose that $-LV<0$ but there is such an interior point $x_0\in{A}$ that

\begin{eqnarray*}
V(x_0)=\max_{x\in{A}}V(x)=M\le0
\end{eqnarray*}
Then $a<x_0<b$  and thus $V_S(x_0)=0$ and $V_{SS}(x_0)\le0$ and we have
\begin{eqnarray*}
\left.-LV=-\frac{\sigma^2}{2}S^2\frac{d^2V}{dS^2}-(r-q)S\frac{dV}{dS}+rV\right|_{S=x_0}=-\frac{\sigma^2}{2}S^2\frac{d^2V}{dS^2}(x_0)+rM\ge0
\end{eqnarray*}
This contradicts the assumption of $-LV<0$.
In the case of $-LV\le0$ ,let $u=V-\varepsilon$ , then $-Lu=-LV-r\varepsilon<0$ and we have $$\sup_{x\in{A}}u(x)=\sup_{x\in{\partial_pA}}u^+(x)$$
Thus we have
\begin{eqnarray*}
\sup_{x\in{A}}V(x)\le\sup_{x\in{A}}u(x)+\varepsilon=\sup_{x\in{\partial_pA}}u^+(x)\le\sup_{x\in{\partial_pA}}V^+(x)+\varepsilon. 
\end{eqnarray*}
Let $\varepsilon\rightarrow{0}$ ,then we proved $\eqref{eq 3}$.(QED)
\end{proof}

\begin{corollary}
(maximum principle for pricing function of perpetual American options)
\\
$V(S)$  the solution of $\eqref{eq 1}$ attains non-negative maximum value at the boundary points of $A=(a,b)$ ,where $A$ is a arbitrary subdomain of ().
\\
Proof.  If $V(S)$ is the solution of $\eqref{eq 1}$,then there would hold $-LV=0$ in the domain which holds $V>\phi$. Thus Theorem 1 leads to the result.
\end{corollary}
\begin{lemma}
Assume that $V(S)$ is the price of perpetual American put options, that is, the solution to \eqref{eq 1}. If $S$ is in the stopping region, that is, $V(S)=(E-S)^+$, then $S\le{min}\{rE/(q)^+,E\}$. That is, the optimal exercise boundary could not greater than  $\min\{rE/(q)^+,E\}$. Here $(q)^+=
\begin{cases}
q,~~q\ge0\\
0,~~q<0
\end{cases}
$
\end{lemma}
\begin{proof} 
First, note that if $V(S)=(E-S)^+$, then $S\le{E}$ [Jia]. So we can rewrite as $V(S)=E-S$  On the other hand, $\eqref{eq 1}$ can be written as $$\min\left\{-LV,~V-\phi\right\}=0$$ 
Since $(V-\phi)(S)=0$, then $-LV|_S=rE-qS\ge0$, thus if $q\ge0$, then $S\le{rE/q}$.(QED)
\end{proof}

\begin{lemma} 
Let $V(S)$ be the price of perpetual American put options.
 \\
(i) If there is such a $S_0>0$ that $V(S_0)=(E-S_0)^+$, then we have  $V(S)=(E-S)^+$ for all $S(<S_0)$.
\\
(ii) If there is $S_1>0$ such that $V(S_1)=(E-S_1)^+$, then $V(S)>(E-S)^+$ for all $S(>S_1)$.
\end{lemma}
\begin{proof}
(i)From Lemma 1, we have $S_0\le{rE/q}$ .
If the conclusion were not true, that is, we assume that
\begin{equation}\label{eq 4}
0<\exists{S<S_0}:V(S)>E-S.
\end{equation}
Let $(a,b)$ be the largest interval where $\eqref{eq 4}$ holds.
Then we have $V(S)=E-S$ when $S=a$ or $S=b$.
From \eqref{eq 4} and $\min\{{-LV,~V-\phi}\}=0$ , we have $-LV=0$ and $V-(E-S)>0$ on the interval $(a,b)$.
Thus we have $$-L(V-(E-S))=L(E-S)=qS-rE<0~(\because~q\ge{0}\Rightarrow{S<S_0\le{rE/q}}).$$
So from Theorem 1, non-negative maximum value of $V-(E-S)$ is attained at the boundary points.
Since $V(a)-\phi{(a)}=0,~V(b)-\phi{(b)}=0$, we have $V(S)-\phi{(S)}\le0$ on the interval $(a,b)$.
This contradicts $\eqref{eq 4}$.

(ii) If $\exists{S_2>S_1}:V(S_2)=(E-S_2)^+$, then from the conclusion of (i) we have $V(S_1,t)=(E-S_1)^+$ which contradicts the assumption.(QED)

From lemma 2 and its corollary, the existence of the optimal exercise boundary is proved.
\end{proof}
\begin{theorem} 
(Existence of the optimal exercise boundary)
Let $V(S)$ be the price function of perpetual American put option. Then there exists an optimal exercise boundary  $S_0$.
\end{theorem}
\begin{proof}
The interval $(E,+\infty)$ is contained to the continuation region. Let $S_0$ be the infimum of the continuation region, then from lemma 2, $(S_0,+\infty)$ is the continuation region and $(0,S_0)$ is the stopping region. (QED)
\end{proof}
\begin{theorem}
(Uniqueness of the solution to variational inequality for the price of perpetual American put option)
\eqref{eq 1} has at most one solution.
\end{theorem}
\begin{proof}
Proof. Let $V_1$ and $V_2$ be the two solutions of $\eqref{eq 1}$ and $S_1,S_2$ be the optimal exercise boundaries of two solutions, respectively.
Without loss of generality, assume that $S_1>S_2$.  
Then we have $V_1=V_2=\phi$ for $S<S_2$ and $LV_1=LV_2=0$ for $S>S_1$. On the other hand, for $S_2<S\le{S_1}$, we have $V_1=\phi,~LV_2=0,~V_2>\phi$. 
Now consider the interval $(S_2,\infty)$. When $S_2<S\le{S_1}$, we have
 \begin{eqnarray*}
-L(V_1-V_2)=-L\phi=-L\phi=-qS+rE\ge0 (\because q\ge0\Rightarrow{S_1}\le{rE/q}).
\end{eqnarray*}
When $S>S_1$, we have $-L(V_1-V_2)=0$. Thus from the maximum principle (Theorem 1), the non-positive minimum value of $V_1-V_2$ in the interval $[S_2,~\infty)$ is attained at the boundary. But we have $(V_1-V_2)(S_2)=0$ and $(V_1-V_2)(\infty)=0$, which means that $(V_1-V_2)\le0$ in $[S_2,~0)$. This contradicts the fact that $(V_1-V_2)(S_1)=\phi{(S_1)}-V_2(S_1)<0$. So we have $S_1=S_2$.
In $(0,~S_1]$, two solutions are equal to $\phi$, thus, two solutions are coincided. 
In $(S_1,~\infty)$, we have $V_1-V_2=0$ and on the boundary $V_1-V_2=0$. From theorem 1, we have $V_1=V_2$.(QED)
\end{proof}
{\\}

{\bf Remark 2.1}
The solution to the free boundary problem of perpetual American option price considered in [jia] satisfies the variational inequality \eqref{eq 1}, so the existence of the solution to \eqref{eq 1} is already known. So the problem of uniqueness and existence of the solution to the variational inequality pricing model of perpetual American put option and the existence of the optimal exercise boundary are completely solved.

\subsection{Properties of the difference equation of the variational inequality for the price of of perpetual American put options}

Using transformation:
\begin{eqnarray*}
u(x)=V(S),  S=e^x,  \varphi=(E-e^x)^+;  S_0=e^c.
\end{eqnarray*}
the variational inequality \eqref{eq 1} becomes the following problem
\begin{equation}\label{eq 5}
\min\left\{\frac{-\sigma^2}{2}\frac{\partial^2{u}}{\partial{x}^2}-(r-q-\frac{\sigma^2}{2})\frac{\partial{u}}{\partial{x}},~u-\varphi\right\}=0,~~x\in(-\infty, +\infty)
\end{equation}
\begin{eqnarray*}
u(-\infty)=E,~~u(+\infty)=0
\end{eqnarray*}
For $\Delta{x}>0$, let $u_j=u(j\Delta{x}+c),~\varphi_j=\varphi(j\Delta{x}+c)$ and use  difference quotients
\begin{eqnarray*}
\left(\frac{\partial{u}}{\partial{x}}\right)_j=\frac{u_{j+1}-u_{j-1}}{2\Delta{x}},~~\left(\frac{\partial^2{u}}{\partial{x}^2}\right)_j=\frac{u_{j+1}-2u_j+u_{j-1}}{2\Delta{x}^2}.
\end{eqnarray*}
in \eqref{eq 5}, then we have the following difference equation:
\begin{equation}\label{eq 6}
\min\left\{-\frac{\sigma^2}{2}\frac{u_{j+1}-2u_j+u_{j-1}}{2\Delta{x}^2}-(r-q-\frac{\sigma^2}{2})\frac{u_{j+1}-u_{j-1}}{2\Delta{x}}+ru_j,~~u_j-\varphi_j\right\}=0,~~~j\in{Z}
\end{equation}
\begin{eqnarray*}
u_j=E,~~j=-\infty:~~u_j=0,~~j=+\infty.
\end{eqnarray*}
Here, let $\alpha=\Delta{t}$ and consider the fact that $\min(A,B)=0\Leftrightarrow\min(\alpha{A},B)=0~~(\alpha>0)$ .
Then we have
\begin{eqnarray*}
&&\min\left\{-\frac{\sigma^2\Delta{t}}{2\Delta{x}^2}(u_{j+1}-2u_j+u_{j-1})-(r-q-\frac{\sigma^2}{2})\frac{\Delta{t}}{2\Delta{x}}(u_{j+1}-u_{j-1})+r\Delta{t}u_j,~~u_j-\varphi_j\right\}=0\\
&&u_j=E,~~j=-\infty:~~u_j=0,~~j=+\infty
\end{eqnarray*}
Consider that
\begin{eqnarray*}
&&-\frac{\sigma^2\Delta{t}}{2\Delta{x}^2}(u_{j+1}-2u_j+u_{j-1})-(r-q-\frac{\sigma^2}{2})\frac{\Delta{t}}{2\Delta{x}}(u_{j+1}-u_{j-1})+r\Delta{t}u_j=
\\
&&=(1+r\Delta{t})u_j-\left(1-\frac{\sigma^2\Delta{t}}{\Delta{x}^2}\right)u_j-\frac{\sigma^2\Delta{t}}{\Delta{x}^2}\left[\left(\frac{1}{2}+(r-q-\frac{\sigma^2}{2})\frac{\Delta{x}}{2\sigma^2}\right)u_{j-1}\right]
\end{eqnarray*}
and denote
$$
w=\frac{\sigma^2\Delta{t}}{\Delta{x}^2},~~a=\frac{1}{2}+(r-q-\frac{\sigma^2}{2})\frac{\Delta{x}}{2\sigma^2},~~\rho=1+r\Delta{t}.
$$
Then we have
\begin{equation*}
\min\left\{\rho{u_j}-(1-w)u_j-w[au_{j+1}+(1-a)u_{j-1}],~~u_j-\varphi_j\right\}=0
\end{equation*}
Here, consider $min(A,B)\Leftrightarrow\min(\alpha{A},B)=0,~~(\alpha>0)$ again, then we have
$$
\min\left\{u_j-\frac{1}{\rho}\left[(1-w_j)u_j-w[au_{j+1}+(1-a)u_{j-1}]\right],~~u_j-\varphi_j\right\}=0
$$
Then, due to $\min(C-A,C-B)=0\Leftrightarrow{C=max(A,B)}$, we have a following difference equation equal to $\eqref{eq 6}$.
\begin{equation}\label{eq 7}
u_j=max\left\{\frac{1}{\rho}[{(1-w)u_j+w[au_{j+1}+(1-a)u_{j-1}]]},~\varphi\right\}~~j\in{Z}
\end{equation}
$$
u_j=E,~j=-\infty;u_j=0,~j=+\infty
$$
which is equivalent to \eqref{eq 6}

Now, we denote the {\bf discrete Black-Scholes operator} as follows:
$$
I_j(U)=(\rho+w-1)U_j-w[aU_{j+1}+(1-a)U_{j-1}],~j\in{Z},~U\in{l_\infty(Z)}
$$
Here $l_\infty(Z)$ is the Banach space of all bounded two sided sequences of real numbers.

\begin{theorem}
(Maximum principle for the discrete Black-Scholes operator) Assume that $\Delta{x}$ is small enough. Then we have the following two facts:
\\
(i)	If on some “interval” $A$ of $Z$, we have $I_j(U)\le0,j\in{A}$, then the maximum value of $U$  cannot be attained at interior of $A$. That is, there does not exist such $j\in{A}$ that $U_{j-1}<U_j$ and $U_j>U_{j+1}$ hold at the same time.
\\
(ii)	If on some “interval” $A$ of $Z$, we have $I_j(U)\ge0,j\in{A}$, then the minimum value of $U$  cannot be attained at interior of $A$. That is, there does not exist such $j\in{A}$ that $U_{j-1}>U_j$ and $U_j<U_{j+1}$ hold at the same time.
\end{theorem}
\begin{proof}
If $\Delta{x}$ is small enough, then $0<a<1$ in $\eqref{eq 7}$.
Suppose that the conclusion of (i) were not true. That is, assume that $\exists{j}\in{A},~U_{j-1}<U_j,~U_{j+1}<U_j$. Then $aU_{j+1}+(1-a)U_{j-1}<U_j$ and $w>0$, so we have
\begin{eqnarray*}
&&w[aU_{j+1}+(1-a)U_{j-1}]<wU_j<wU_j+r\Delta{t}U_j\Rightarrow
\\
&&\Rightarrow(w+r\Delta{t})U_j-w[aU_{j+1}+(1-a)U_{j-1}]>0\Rightarrow
\\
&&\Rightarrow{I_j(U)}=(\rho+w-1)U_j-w[aU_{j+1}+(1-a)U_{j-1}]>0
\end{eqnarray*}
This contradicts $I_j(U)<0,~j\in{A}$. Thus, we proved the conclusion of (i).
The conclusion of (ii) is derived by applying the result of (i) to $-U$.(QED)
\end{proof}
{\\}

By using this theorem, we can prove the maximum principle of the solution to difference equation \eqref{eq 7} for the variational inequality for perpetual American put option price.

\begin{theorem}
(Maximum principle for the difference equation for the variational inequality for perpetual American put option price) 
Let $U={U_j}\in{Z}$ be the solution of $\eqref{eq 7}$. Then, the maximum and minimum value of $U$ on any subinterval $A=[j_1,j_2]$ of $\sum_1=\{j\in{Z}|U_j>\varphi_j\}$ is attained at the boundary. Thus $U_j$ is the monotone with respect to j on $\sum_1$.
\end{theorem}
\begin{proof}
Since $U$ is the solution of $\eqref{eq 7}$ and $U_j>\varphi_j$ on $\sum_1$, we have $I_j(U)=0$.
Thus, from theorem 4, we have the desired result. (QED)
\end{proof}
\begin{lemma}
Suppose that $E>0,~r>0$ and $\Delta{x}$ is sufficiently small.
Let $U=\{U_j\}_{j\in{Z}}$ be the solution to $\eqref{eq 7}$.
Denote by 
$$
j^*=\min\{j\in{Z}:E\le{e^{j\Delta{x}+c}}\},~~j^{**}=\min\{j\in{Z}:rE\le{qe^{j\Delta{x}+c}}\}
$$
Then, we have
$$
U_j=\varphi\Rightarrow{j<inf\{j^*,~j^{**}\}}<\infty
$$
\end{lemma}
\begin{proof}
If $E>0$, then $-\infty<j^*<\infty$ and $U_j>\varphi=0$ for $j\ge{j^*}$. If $U_j=\varphi_j$, then $j\le{j^*-1}$ and here $\varphi=E-e^{j\Delta{x}+c}$. On the other hand, if $j\le{j^*-1}$ and if $U_j=\varphi_j$, then we have
$$
\varphi_j=U_j>\frac{1}{\rho}[(1-w)U_j+w[aU_{j+1}+(1-a)U_{j-1}]]
$$ 
$$
\ge\frac{1}{\rho}[(1-w)\varphi_j+w[a\varphi_{j+1}+(1-a)\varphi_{j-1}]]
$$
This means that $I_j(\varphi)>0$. Since
$$
ae^{\Delta{x}}+(1-a)e^{-\Delta{x}}=1+(r-q)\frac{\Delta{x}^2}{\sigma^2}+O(\Delta{x}^4)
$$
we have
\begin{eqnarray*}
I_j(\varphi)&=&(\rho+w-1)\varphi_j-w(a\varphi_{j+1}+(1-a)\varphi_{j-1})=
\end{eqnarray*}

\begin{equation}\label{eq 8}
=\Delta{t}\left((rE-qe^{j\Delta{x}+c})+\frac{\sigma^2}{\Delta{x}^2}\right)>0
\end{equation}
Thus, if $\Delta{x}$ is sufficiently small, we have $rE>qe^{j\Delta{x}+c}$ and $j<j^{**}$.(QED)
\end{proof}
{\\}

{\bf Remark2.}
If $r=0,~q\ge0$, then $j^{**}=-\infty$ and the exercise region is empty, so the exercise boundary does not exist. If $r=0,~q<0$ or $r>0,~q\le0$, then $j^{**}=\infty$ and exercise boundary is $j^*$.
\begin{lemma}
Assume that $E>0,r>0$. Let $U=\{U_j\}_j\in{Z}$ be the solution of $\eqref{eq 7}$ and $\Delta{x}$ is sufficiently small. Then, we have
$$
(i)~~\exists{j_0}\in{Z}:U_{j_0}=\varphi_{j_0}\Rightarrow\forall<j_0,~~U_j=\varphi_j
$$
$$
(ii)~~\exists{j_1}\in{Z}:U_{j_1}>\varphi_{j_1}\Rightarrow\forall>j_1,~~U_j>\varphi_j
$$
\end{lemma}
\begin{proof}
If (i) were not true, there would exist such a set
$$
A=\{j\in{Z}:-\infty\le{j_1}\le{j}\le{j_2}\le\infty\}
$$  
that 
$$
j_1=-\infty\Rightarrow{U_{j_1}}=E;~j_1\neq{-\infty}\Rightarrow{U_{j_1}=\varphi_{j_1}};j_1<j<j_2\Rightarrow{U_j>\varphi_j}~~~(*)
$$ 
Without loss of generality, assume that $j_2=j_0$, then we have $U_{j_0}=\varphi_{j_0}$.
Then, from lemma 3, we have $j_0<min\{j^*,j^{**}\}$.
$$
U_j>\varphi_j\Rightarrow{U_j}=\frac{1}{\rho}\{(1-2)U_j+w[aU_{j+1}+(1-a)U_{j-1}]\}\Rightarrow{I_j(U)}=0
$$
We have $rE>qe^{j\Delta{x}+c}$ for $j<j_0$, and thus from $\eqref{eq 8}$ we have $I_j(\varphi)>0$ .
,that is,
$$
I_j(\varphi)>0,\forall{j}\in{A}. 
$$
Thus, we have 
$$
I_j(\varphi)>0,~I_j(U)=0, \forall{j}\in{A}
$$
and thus we have
$$
I_j(U-\varphi)<0,~j\in{A};U_{j_0}=\varphi_{j_0},U_{j_1}=\varphi_{j_1}
$$
From theorem 5, the maximum value of $U-\varphi$ in A is attained at the boundary points $j=j_0,~j=j_1$.
Therefore we have 
$$
U_j\le{\varphi_j}~~~j\in{A},j\neq{j_0,j_1}
$$ 
This contradicts the assumption (*).
If (ii) were not true, then we have $\exists{j_3}>j_1\Rightarrow{U_{j_3}=\varphi_{j_3}}$ and from the result of (i) we have $U_{j_1}=\varphi_{j_1}$, which contradicts the assumption. (QED)
\end{proof}

\begin{theorem}
(Existence of optimal exercise boundary of explicit difference scheme) 
Assume that $U_{j_1}=\varphi_{j_1}$ and $U=\{U_j\}_{j\in{Z}}$ be the solution of $\eqref{eq 7}$.
Then there exists such $j_0\in{Z}$ that 
$$
j<j_0\Rightarrow{U_j=\varphi_j} ~and~ j\ge{j_0}\Rightarrow{U_j>\varphi_j}. 
$$
\end{theorem}

\begin{proof}
From lemma 3, if $j\ge{\min(j^*,j^{**})}$, then $U_j>\varphi_j$. Denote $j_0=min\{j:U_j>\varphi_j\}$. Then from lemma 4 we have $j\ge{j_0}\Rightarrow{U_j>\varphi_j}$ and $j<j_0\Rightarrow{U_j=\varphi_j}$.
\end{proof}

In what follows, we denote optimal exercise boundary by $j_*=j_0-1$. From the result of lemma 4, the explicit difference scheme (7) can be written as the following difference equation:
\begin{eqnarray}
	&u_j=\frac{1}{\rho}\{(1-w)u_j+w[u_{j+1}+(1-a)u_{j-1}]\},~~~~~~~~~~~~~~~~~ j>j_*\label{eq 9}\\
	&u_{j_{*}}=\varphi_{j_{*}}\ge\frac{1}{\rho}\{(1-w)u_{j_{*}}+w[u_{j_{*}+1}+(1-a)u_{j_{*}-1}]\}, ~~~~j=j_{*}\label{eq 10}\\
	&u_j=0, ~~~~~~~~~~~~~~~~~~~~~~~~~~~~~~~~~~~~~~~~~~~~~~~~~~~~~~~~~~~~~~~~~j=+\infty\label{eq 11}
\end{eqnarray}
\begin{theorem}
(Uniqueness of the solution of the explicit difference scheme of perpetual American options)
The solution of the equation \eqref{eq 9}, \eqref{eq 10}, \eqref{eq 11} is unique if it exists.
\end{theorem}
\begin{proof}
Assume that we have two different solutions $(U^1,j_*^1),(U^2,j_*^2)$. Without loss of generality  $j_*^1>j_*^2$ . Then we have
\begin{equation*}
\begin{cases}
	I_{j}(U^1-U^2)
		\begin{cases}
			=0,~~~~~~{j}\ge{j_*^1},\\
			>0,~~~~~~j_*^2\le{j}<j_*^1,
		\end{cases}\\
U_{+\infty}^1=U_{+\infty}^2=0,~U_{j_*^2}^1=U_{j_*^2}^2=\varphi_{j_*^2}
\end{cases}
\end{equation*}
{\\}
Thus we have $I_{j}(U^1-U^2)$ for all $\forall{j}\in(j_*^2,~+\infty)$ and by the maximum principle, the non-positive minimum value of $U^1-U^2$ is attained at the boundary of $(j_*^2,+\infty)$. 
{\\}
Since $U^1-U^2|_{j=+\infty}$, we have $U^1-U^2>0$,$\forall{j}\in(j_*^2,+\infty)$, that is, $U^1>U^2,\forall{j}\in(j_*^2,+\infty)$ and this contradicts the fact that $U_{j_*^1}^1=\varphi_{j_*^1}<U_{j_*^1}^2$.
{\\}
Thus we have $j_*^1=j_*^2$. Therefore, the two solutions are equal to $\varphi_j$ in $(-\infty,j_*^1]$. In $(j_*^1,\infty)$ we have $I_j(U^1-U^2)=0$ and $U^1-U^2=0$ at the boundary, so from Theorem $\eqref{theorem4}$, we have $U^1=U^2$. (QED)
\end{proof}
{\\}

Now, we find a solution $(U, j^*)$ to the equation $\eqref{eq 9}$, $\eqref{eq 10}$, $\eqref{eq 11}$. To do this, it
is sufficient to find the solution satisfying the equality in $\eqref{eq 10}$. Substitute $U_j=\xi^j$
into $\eqref{eq 9}$ to obtain the characteristic equation:
$$
wa\xi^2-(\rho+w-1)\xi+w(1-a)=0.
$$
This equation has two roots:
$$
\xi_{1,2}=\frac{\rho+w-1\pm\sqrt{(\rho+w-1)^2-4aw^2(1-a)}}{2aw}
$$
Obviously, $\xi_1<1<\xi_2$, so the general solution has the form $U_j=c_1\xi_1^j+c_2\xi_2^j$ with constants $c_1,c_2$. Since $\xi_1^j\rightarrow{0}$ when $j\rightarrow+\infty, ~c_2$ must be zero in order to satisfy the boundary condition $\eqref{eq 11}$ $(U_j=0,~j=+\infty)$.
{\\}
Thus, when $j>j_*$, the solution of $\eqref{eq 9}$ has the form
\begin{equation}\label{eq 12}
U_j=c_1\xi_1^j.
\end{equation}
If $j\le{j_*}$, then we have $U_j=\varphi_j=E-e^{j\Delta{x}+c}$. In particular, when $j=j_*$, we find $U_{j_*}$ so that $U_{j_*}$ satisfies the equality in $\eqref{eq 10}$. Then $U_{j_*}$ becomes the solution to $\eqref{eq 9}$ and thus $U_{j_*}$ have to be expressed as follows:
\begin{equation}\label{eq 13}
U_{j_*}=c_1\xi_1^{j_*}=\varphi_{j_*}=E-e^{j_*\Delta{x}+c}.
\end{equation}
From this, the constant $c_1$ can be expressed as follows:
$$
c_1=(E-e^{j_*\Delta{x}+c})\xi_1^{-j_*}
$$
Now, find $j_*$. From $\eqref{eq 12}$, $\eqref{eq 13}$, we have
$$
U_{j_*+1}=c_1\xi_1^{j_*+1}=c_1\xi_1^{j_*}\xi_1=\xi_1(E-e^{j_*\Delta{x}+c}),~~U_{j_*-1}=E-e^{(j_*-1)\Delta{x}+c}
$$
From $\eqref{eq 10}$ we have
\begin{eqnarray*}
U_{j_*}&=&\varphi_{j_*}=E-e^{j_*\Delta{x}+c}=\frac{1}{\rho}\left\{(1-w)U_{j_*}+w[aU_{j_*+1}+(1-a)U_{j_*-1}]\right\}=\\
&=&\frac{1}{\rho}\left\{(1-w)(E-e^{j_*\Delta{x}+c})+w[a\xi_1(E-e^{j_*\Delta{x}+c})+(1-a)(E-e^{(j_*-1)\Delta{x}+c})]\right\}\\
\end{eqnarray*}
\begin{align*}
&\Leftrightarrow\rho(E-e^{j_*\Delta{x}+c})=(1-w)(E-e^{j_*\Delta{x}+c})+w[a\xi_1(E-e^{j_*\Delta{x}+c})+(1-a)(E-e^{j_*\Delta{x}+c}/e^{\Delta{x}})]\\
&\Leftrightarrow\rho{E}-(1-w)E-wa\xi_1{E}-w(1-a)E=[\rho-(1-w)-wa\xi_1-w(1-a)e^{-\Delta{x}}]e^{j_*\Delta{x}+c}\\
&\Leftrightarrow{e^{j_*\Delta{x}+c}}=\frac{\rho-[(1-w)+w(a\xi_1+(1-a))]E}{\rho-\{(1-w)-w[a\xi_1+(1-a)e^{-\Delta{x}}]\}}\\
&\Leftrightarrow{j_*}=\left[\frac{1}{\Delta{x}}\left(log\frac{\{\rho-[(1-w)+w(a\xi_1+(1-a))]E}{\rho-\{(1-w)+w[a\xi_1+(1-a)e^{-\Delta{x}}]\}}-c\right)\right].
\end{align*}
Now, we denote
\begin{equation}\label{eq 14}
f=\frac{1}{\Delta{x}}\left(log\frac{\{\rho-[(1-w)+w(a\xi_1+(1-a))]E}{\rho-\{(1-w)+w[a\xi_1+(1-a)e^{-\Delta{x}}]\}}-c\right).
\end{equation}
Then we define $(U,j_*)$ as follows:
\begin{equation}\label{eq 15}
j_*=[log ~f],~~~U_j=\begin{cases}
					(E-e^{j_*\Delta{x}+c})\xi_1^{j-j_*},~~~~~j>j_*\\
E-e^{j_*\Delta{x}+c}~~~~~~~~~~~~~~~j\le{j_*}
\end{cases}
\end{equation}
Then, $(U,j_*)$ is the solution to \eqref{eq 9}, \eqref{eq 10}, \eqref{eq 11}. Thus we proved the existence theorem of the solution of difference equation of perpetual American put option.
\begin{theorem}
The solution of explicit difference scheme \eqref{eq 9}, \eqref{eq 10}, \eqref{eq 11} of perpetual American put option price exists and it is expressed as \eqref{eq 15}.
\end{theorem}
\begin{corollary}
If $r=0$ in \eqref{eq 14}, then  $\rho=1,\xi_1=1$, thus $f=0$ and optimal exercise boundary does not exist.
\end{corollary}
{\bf Remark 3}. By Theorem 6, 7, 8, the problem of uniqueness and existence of the solution and optimal exercise boundary of explicit difference scheme for perpetual American put option price is solved. From the consideration in \cite{JD}, \cite{Jia}, the binomial tree method of \cite{LL} can be considered as a special difference equation for the variational inequality neglecting an infinitesimal of the same order as $\Delta{x}^3$ and thus such study can be viewed as an extension of the results of \cite{LL}.

\subsection{Convergence of the approximated solution }
\indent
Consider the concept of viscosity solutions to variational inequality of perpetual American options.

If $u\in USC(\mathbf{R})~ (LSC(\mathbf{R}))$ satisfies the following two conditions, then $u$ is called the {\it viscosity subsolution (supersolution)} of the variational inequality \eqref{eq 4}: 

(i) $u(-\infty)\le (\ge) E$, $u(+\infty)\le(\ge){0}$.

(ii)  If $\Phi\in C^2( \mathbf{R})$ and $u-\Phi$ attains its local maximum (minimum) at $x \in\mathbf{R}$, we have 
$$\min\left\{-\frac{\sigma^{2}}{2}\frac{\partial^2 \Phi}{\partial x^2}-\left(r-q-\frac{\sigma^{2}}{2}\right)\frac{\partial \Phi}{\partial x}+ru,u-\varphi\right\}_{x}\le (\ge) 0$$
$u\in{C(\mathbf{R})}$ is called the viscosity solution of the variational inequality $\eqref{eq 4}$ if it is both viscosity subsolution and viscosity supersolution of $\eqref{eq 4}$.
For $x\in[(j-1/2)\Delta{x}+c,(j+1/2)\Delta{x}]$, we define the extension function $u_{\Delta{x}}(x)$ as follows:
$$
u_{\Delta{x}}(x):=U_j
$$
Here $U_j$ is the solution of $\eqref{eq 7}$ and given by $\eqref{eq 15}$.
{\\}
In $l^{\infty}(Z)$, we define $(U_j)\le(V_j)\Leftrightarrow{U_j}\le{V_j},~\forall{j}\in{Z}$ .  We define the operator ${\mathbf{F}}$ in $l^{\infty}(Z)$ as follows:
$$
[\mathbf{F}(\mathbf{U})]_{j}=max\left\{\frac{1}{\rho}\left\{(1-w)U_j+w[aU_{j+1}+(1-a)U_{j-1}]\right\},\varphi_j\right\},~~~~\mathbf{U}=(U_j)\in{l^\infty}(Z).
$$
Then the solution $(U_j)$ of \eqref{eq 7} is a fixed point of $\mathbf{F}$.
\begin{lemma}
If $0<w\le{1}$, $\left|\frac{\Delta{x}}{\sigma^2}\left(r-q-\frac{\sigma^2}{2}\right)\right|<1$, then, $\mathbf{F}$ is monotone in $l^{\infty}(Z)$.\\
That is, we have 
\end{lemma}
$$
\mathbf{U}\le\mathbf{V},~~~\mathbf{U},\mathbf{V}\in{l^{\infty}(Z)}\Rightarrow\mathbf{FU}\le\mathbf{FV}.
$$
\begin{proof}
Noting that from the assumption we have $1-\alpha\ge0$, $0<w<1$, the required result easily comes. (QED)
\end{proof}
\begin{lemma}
If $\mathbf{U}\in{l^{\infty}(Z)}, K\ge{0} ,\mathbf{K}=(...K,K,K...)$  . Then, we have
\end{lemma}
$$
\mathbf{F}(\mathbf{U}+\mathbf{K})\le\mathbf{FU}+\mathbf{K}
$$
\begin{proof}
Since $\rho>1$, we have
\begin{eqnarray*}
\mathbf{F}(\mathbf{U}+\mathbf{K})&=&\\
&=&\left(max\left\{\frac{1}{\rho}\left[(1-w)(U_j+K)+w\left(\alpha(U_{j+1}+K)+(1-\alpha)(U_{j-1}+K)\right)\right],\varphi\right\}\right)_{j=-\infty}^\infty\\
&\le&\left(max\left\{\frac{1}{\rho}\left[(1-w)U_j+w(\alpha{U_{j+1}}+(1-\alpha)U_{j-1}))\right]+\frac{K}{\rho}\right\}\right)_{j=-\infty}^{\infty}\\
&\le&\mathbf{FU}+\mathbf{K}. (\bf{QED})
\end{eqnarray*}
\end{proof}
\begin{lemma}
The price $U,~j\in{\mathbf{Z}}$ of explicit difference scheme of perpetual American put options is bounded.
\end{lemma}
\begin{proof}
 Due to $\eqref{eq 15}$ and $0<\xi_1<1$, we have $\forall{j},~0\le{U_j}\le{E}$. (QED)
\end{proof}
\begin{theorem}
Let $u(x)$ be the viscosity solution of \eqref{eq 4}. If $\left|r-q-\frac{\sigma^2}{2}\right|{\frac{\Delta{x}}{\sigma^2}}\le{1}$, then $u_{\Delta{x}}(x)$converges to $u(x)$ as $\Delta{x}\rightarrow{0}$.
\end{theorem}
\begin{proof}
Denote
\begin{eqnarray*}
u^*(x)=\lim_{\Delta{x}\rightarrow{0}}\sup_{y\rightarrow{x}}u_{\Delta{x}}(y),~~u_*(x)=\lim_{\Delta{x}\rightarrow{0}}\inf_{y\rightarrow{x}}u_{\Delta{x}}(y)
\end{eqnarray*}
From Lemma 7, $u^*$ and $u_*$ are well-defined and we have $0\le{u_*(x)}\le{u^*(x)}\le{E}$. Obviously, $u^*\in{USC(\mathbf{R})}$ and $u^*\in{LSC(\mathbf{R})}$. If we prove that $u^*$ is the subsolution and $u_*$ the supersolution of $\eqref{eq 4}$, then we have $u^*\le{u_*}$ and thus $u^*=u_*=u(x)$ becomes the viscosity solution of $\eqref{eq 4}$, and therefore we have the convergence of the approximate solution $u_{\Delta{x}}(x)$.
We will prove that $u^*$ is the subsolution of $\eqref{eq 4}$. (The fact that $u^*$ is the supersolution is similarly proved.) Suppose that $\phi\in{C^2({\mathbf{R}})}$ and $u^*-\phi$ attains a local maximum at $x_0\in\mathbf{R}$. We might as well assume that $(u^*-\phi)(x_0)=0$ and $x_0$ is a strict local maximum on $B_r=\{x:\left|x-x_0\right|\le{r}\},~r>0$. Let $\Phi=\phi-\varepsilon.~~\varepsilon>0$, then $u^*-\Phi$ attains a strict local maximum at $x_0$ and
\begin{equation}\label{eq 16}
(u^*-\Phi)(x_0)>0
\end{equation}
From the definition of $u^*$, there exists a sequence $u_{\Delta{x}_k}(y_k)$ such that $\Delta{x}_k\rightarrow{0}$, $y_k\rightarrow{x_0}$ and 
\begin{equation}\label{eq 17}
\lim_{k\rightarrow\infty}u_{\Delta{x}_k}(s_k,y_k)=u^*(t_0,x_0)
\end{equation} 
If we denote the global maximum point of $u_{\Delta{x}_k}-\Phi$ on $B_r$ by $\hat{y}_k$, then there exists a subsequence $u_\Delta{x}_{k_i}$ such that when $k_i\rightarrow\infty$, we have 
\begin{equation}\label{eq 18}
\Delta{x}_{k_i}\rightarrow{0},\hat{y}_{k_i}\rightarrow{x_0}~~and~~ \left(u_\Delta{x}_{k_i}-\Phi\right)(\hat{y}_{k_i})\rightarrow\left(u^*-\Phi\right)(x_0)
\end{equation}
Indeed, suppose $\hat{y}_{k_i}\rightarrow\hat{y}$, then from $\eqref{eq 17}$ we have
\begin{eqnarray*}
\left(u^*-\Phi\right)(x_0)=\lim_{{k_i}\rightarrow\infty}\left(u_{\Delta{x_{k_i}}}-\Phi\right)(y_k)\le\lim_{k_i\rightarrow\infty}\left(u_{\Delta{x}_{k_i}}-\Phi\right)(\hat{y}_{k_i})\le\left(u^*-\Phi\right)(\hat{y}).
\end{eqnarray*}
Therefore we have $\hat{y}=x_0$, since $x_0$ is a strict local maximum of $\left(u^*-\Phi\right)$. 
Thus for sufficiently large $k_i$, if $(\hat{S}_{k_i}+\Delta{t}(\Delta{x_{k_i}}))\in{B_r}$, then we have 
\begin{eqnarray*}
\left(u_{\Delta{x}_{k_i}}-\Phi\right)(x)\le\left(u_{\Delta{x}_{k_i}}-\Phi\right)(\hat{y}_{k_i})
\end{eqnarray*}
that is,
\begin{equation}\label{eq 19}
u_{\Delta{t}_{k_i}}(x)\le\Phi(x)+\left(u_{\Delta{x}_{k_i}}-\Phi\right)(\hat{y}_{k_i})
\end{equation}
From $\eqref{eq 16}$ and $\eqref{eq 18}$, when $k_i$ is sufficiently large, then we have 
\begin{equation}\label{eq 20}
\left(u_{\Delta{x}_{k_i}}-\Phi\right)(\hat{y}_{k_i})>0.
\end{equation}
For every $k_i$, $j_{k_i}=j$ select  such that $\hat{S}_{k_i}\in[t_n,t_{n+1})$, $\hat{y}_{k_i}\in[(j-1/2)\Delta{x}_{k_i}+c,(j+1/2)\Delta{x}_{k_i}+c)$. Then from $\eqref{eq 19}$, Lemma 5 and Lemma 6 we have
\begin{eqnarray*}
u_{\Delta{x}_{k_i}}(\hat{y}_{k_i})&=&U_j=(\mathbf{FU})_j=[\mathbf{Fu}_{\Delta{x}_{k_i}}(\bullet)](\hat{y}_{k_i})\\
&\le&\left\{\mathbf{F}[\Phi(\bullet)+(u_{\Delta{x}_{k_i}}-\Phi)(\hat{y}_{k_i})]\right\}(\hat{y}_{k_i}).\\
&\le&\left\{\mathbf{F}[\Phi(\bullet)]\right\}(\hat{y}_{k_i})+(u_{\Delta{x}_{k_i}}-\Phi)(\hat{y}_{k_i}).
\end{eqnarray*}
Thus we have
\begin{eqnarray*}
\Phi(\hat{y}_{k_i})-\left\{\mathbf{F}[\Phi(\bullet)]\right\}(\hat{y}_{k_i})\le{0}.
\end{eqnarray*}
Therefore using $\eqref{eq 7}$ we have
\begin{align*}
\Phi&(\hat{y}_{k_i})-\left\{\mathbf{F}[\Phi(\bullet)]\right\}(\hat{y}_{k_i})=\\
&=\Phi(\hat{y}_{k_i})-max\left\{\frac{1}{\rho}\left\{(1-w)\Phi_{j}+w[a\Phi_{j+1}+(1-a)\Phi_{j-1}]\right\},~\varphi_j\right\}\\
&=\Phi(\hat{y}_{k_i})-max\left\{\frac{1}{1+r\Delta{t}}\left[\left(1-\frac{\sigma^2\Delta{t}}{\Delta{x_{k_i}}}\right)\Phi(\hat{y}_{k_i})+\frac{\sigma^2\Delta{t}}{\Delta{x_{k_i}}}\left[\left(\frac{1}{2}+(r-q-\frac{\sigma^2}{2})\frac{\Delta{x}_{k_i}}{2\sigma^2}\right)\Phi(\hat{y}_{k_i}+\Delta{x}_{k_i})+\right.\right.\right.\\
&\left.\left.\left.\left(\frac{1}{2}
-(r-q-\frac{\sigma^2}{2})\frac{\Delta{x}_{k_i}}{2\sigma^2}\right)\Phi(\hat{y}_{k_i}-\Delta{x}_{k_i})\right]\right],~\varphi_j\right\}\le0.\\
\end{align*}
This inequality is equivalent to the following.
\begin{align*}
\min&\left\{\frac{\Delta{t}}{1+r\Delta{t}}\left[\frac{\Phi(\hat{y}_{k_i})-\Phi(\hat{y}_{k_i})}{\Delta{t}}\right.\right.-\\
&-\frac{\sigma^2}{2}\frac{\Phi(\hat{y}_{k_i}+\Delta{x_{k_i}})-2\Phi(\hat{y}_{k_i})+\Phi(\hat{y}_{k_i}-\Delta{x_{k_i}})}{2\Delta{x_{k_i}}}\\
&\left.\left.-\left(r-q-\frac{\sigma^2}{2}\right)\frac{\Phi(\hat{y}_{k_i}+\Delta{x_{k_i}})-\Phi(\hat{y}_{k_i}-\Delta{x_{k_i}})}{2\Delta{x_{k_i}}}+r\Phi(\hat{y}_{k_i})\right],~~~\Phi(\hat{y}_{k_i})-\varphi_j\right\}\le0.
\end{align*}
Noting that $\frac{\Delta{t}}{1+r\Delta{t}}$, we have
\begin{align*}
\min&\left\{-\frac{\sigma^2}{2}\frac{\Phi(\hat{y}_{k_i}+\Delta{x_{k_i}})-2\Phi(\hat{y}_{k_i})+\Phi(\hat{y}_{k_i}-\Delta{x_{k_i}})}{2\Delta{x_{k_i}}}\right.\\
&\left.-\left(r-q-\frac{\sigma^2}{2}\right)\frac{\Phi(\hat{y}_{k_i}+\Delta{x_{k_i}})-\Phi(\hat{y}_{k_i}-\Delta{x_{k_i}})}{2\Delta{x_{k_i}}}+r\Phi(\hat{y}_{k_i}),~~~\Phi(\hat{y}_{k_i})-\varphi_j\right\}\le0.
\end{align*}
Let $k_i\rightarrow\infty$, then $\Delta{x_{k_i}}\rightarrow0$ and we have
$$
\min\left\{-\frac{\sigma^2}{2}\frac{\partial^2\Phi}{\partial{x}^2}-\left(r-q-\frac{\sigma^2}{2}\right)\frac{\partial\Phi}{\partial{x}}+r\Phi,~~\Phi-\varphi\right\}_{x_0}\le0.
$$
(Here we considered $\hat{y_{k_i}}\rightarrow{x_0},~~\varphi_{k_i}\rightarrow\varphi(x_0)$.) Here let $\varepsilon\rightarrow0$ then we have
$$
\min\left\{-\frac{\sigma^2}{2}\frac{\partial^2\phi}{\partial{x^2}}-\left(r-q-\frac{\sigma^2}{2}\right)\frac{\partial\phi}{\partial{x}}+r\phi,~\phi-\varphi\right\}_{x_0}\le0.
$$
Since $u^*(x_0)=\phi(x_0)$, $u^*$ is a subsolution of $\eqref{eq 18}$. Thus we proved theorem. (QED)
\end{proof}
{\\}

{\bf Remark.} The result of this section can be viewed as an extension of the results of [JD] to the case with the expiry time $T=\infty$

\section{The Limit of the Price of American Option with respect to Maturity}
\indent
Naturally, the price of perpetual American option can be seen as a limit of the price of American Option when the maturity goes to infinity. This approach excludes the apriori assumtion that the price of perpetual American option does not depend on time. In this section we show that the limits of the prices of variational inequality and BTM models for American Option when the maturity goes to infinity do not depend on time.

\subsection{The Limit of the Solution to a Variational Inequality Model}

\indent

When $r>0,~q$ and $\sigma>0$ are constants, the Black-Scholes operator is defined as follows:
$$
-LV=-\frac{\partial{V}}{\partial{t}}-\frac{\sigma^2}{2}S^2\frac{\partial^2{V}}{\partial{S^2}}-(r-q)S\frac{\partial{V}}{\partial{S}}+rV
$$

Let $V(S,t;T)$ be the price of American put option with the maturity $T$, that is, the solution to the variational inequality 
\begin{eqnarray}\label{eq 21}
&&\min\{-LV,V-\varphi\}=0,~~~~~~~~~~~~~0<t<T,~S>0\nonumber\\
&&V(S,T)=\varphi(S)=(E-S)^+,~~~~~S>0
\end{eqnarray}

Then the solution $V(S,t,T)$ uniquely exists and $V(S,t,T)$ is decreasing with respect to $S$ and $t$ and increasing on $T$. Furthermore, $V(S,t,T)$ is bounded. That is, we have 
\begin{equation}\label{eq 22}
0\le{V(S,t,T)}\le{E},~~~0<S<\infty,~~0\le{t}\le{T}
\end{equation}
Thus, the following limit exists.
\begin{equation}\label{eq 23}
U(S,t)=\lim_{T\rightarrow\infty}V(S,t,T)=sup_{T}V(S,t,T),~~0<S<\infty,~~0\le{t}<\infty.
\end{equation}
$U(S,t)$ is decreasing on $S$ and $t$ and it can be seen as "the price of perpetual American put option" but this function seems to be time dependent. 

On the other hand, it is natural to think that $U(S,t)$ is the solution to the following problem.
\begin{eqnarray}\label{eq 24}
&&\min\{-LU,~~U-\phi\}=0,~~~~~~~~~~~~~t<0,~S>0\nonumber\\
&&V(0+,t)=E,~~V(+\infty,t)=0,~~~~~S>0.
\end{eqnarray}

Thus $\eqref{eq 24}$ can be seen as "the pricing model of perpetual American put option". it is different from the model of \cite{Jia} derived under the assumption that the price of perpetual American put option is independent on time. This variational inequality does not belong to the range of application of a general theory of existence and uniqueness of solution to variational inequalities discussed in \cite{Fri} or \cite{CIL} because the spatial variable and time variable intervals are both infinite intervals. The price function of perpetual American option consructed in \cite{Jia} satisfies $\eqref{eq 24}$. So if we can prove the uniquness of the solution to $\eqref{eq 24}$, then we will have the independence on time variable of the solution to $\eqref{eq 24}$. Thus we will proce the uniquness of the solution to $\eqref{eq 24}$. 

If $U(S,t)$ is the solution to $\eqref{eq 24}$, then we have $-LU>0$ on the region $\Sigma_1$(exercise region or stopping region) where $U(S,t)=(E-S)^+$ and $-LU=0$ on the region $\Sigma_2$(the continuation region) where $U(S,t)>(E-S)^+$. The solution to $\eqref{eq 24}$ is always nonnegative.

The parabolic boundary of the region $A=(a,b)\times(0,T)(0\le{a}<b\le\infty,T>0)$ is defined as follows:
$$
\partial_p{A}=\{a\}\times(0,T)\cup\{b\}\times(0,T)\cup(a,b)\times\{T\}.
$$

\begin{theorem}
(maximum principle of the Black-Scholes differential operator)\\
1) For $V(S,t)\in{C^{2,1}(A)}$, if $-LV<(>0) (S,t)\in{A}$, the nonnegative maximum (nonpositive minimum) value of $V$ cannot be attained at the parabolic boundary of $A$. Furthermore, if $-LV\le(\ge)0 (S,t)\in{A}$, then we have 
\begin{equation}\label{eq 25}
\sup_{x\in{A}}V(x)=\sup_{x\in\partial_p{A}}V^+(x)~~~(\inf_{x\in{A}}V(x)=\inf_{x\in\partial_p{A}}V^-(x))
\end{equation}\\
2) Fix $t>0$. Let $A_t=\{(S,t)\in{A}\}=(a,b)\times\{t\}$. If $-LV<0,~~(S,t)\in{A_t}$ and  $V_t\le0$ ($V$ decreasing on $t$), then we have 
\begin{equation*}
\sup{V(x)}=max\{V^+(a,t),V^+(b,t)\}.
\end{equation*}
\end{theorem}

\begin{proof}
1) In fact, suppose that it attains the nonnegative maximum value at the parabolic interior though $-LV<0$, there exists such a point $x_0=(S_0,t_0)$ that 
$$
V(x_0)=\max_{x\in{A}}V(x)=M\ge0
$$
Then, we have $a<S_0<b,~~0\le{t_0}<T$. If $a<S_0<b,~~0<t_0<T$, then we have 
$$
V_t(x_0)=V_S(x_0)=0,~~V_{SS}(x_0)\le0
$$
Then, we have 
$$
-LV(x_0)=-V_t-\frac{\sigma^2}{2}S^2V_{SS}-(r-q)SV_S+rV|_{x_0}\ge0
$$
This contradicts $-LV<0$. If $a<S_0<b,~t_0=0$, then we have
$$
-V_t(x_0)\ge0,~~V_S(x_0)=0,~~V_{SS}(x_0)\le0
$$
Thus, we have $-LV(x_0)\ge0$. This contradicts $-LV<0$.\\
If $-LV\le0$, let $u=V-\varepsilon$ then have $-Lu=-LV-r\varepsilon<0$, thus we have
$$
\sup_{x\in{A}}u(x)=\sup_{x\in{\partial_p{A}}}u^+(x).
$$ 
Thus we have  
$$
\sup_{x\in{A}}V(x)\le{\sup_{x\in{A}}}u(x)+\varepsilon=\sup_{x\in{\partial_p{A}}}u^+(x)+\varepsilon\le\sup_{x\in{\partial_p{A}}}V^+(x)+\varepsilon.
$$
and here let $\varepsilon\rightarrow0$, then we have $\eqref{eq 25}$.\\
2) If $a<S_0<b$ and $V(S_0,t)=max_{A_t}V(S,t)=M\ge0$, then we have $V_S(S_0,t)=0,V_{SS}(S_0,t)\le0$, and from decreasing on $t$, we have $-V_t(S_0,t)\ge0$. Thus, we have $-LV(S_0,t)\ge0$. (QED)
\end{proof}
\begin{lemma}
Assume that $V(S,t)$ is the solution to $\eqref{eq 24}$ and $V(S,t)=(E-S)^+$. If $q\ge0$, then we have $S\le\min\{rE/q,E\}$. If $q<0$, then we have $S\le{E}$.
\end{lemma}

\begin{proof}
First, note that if $V(S,t)=(E-S)^+$ then $S\le{E}$[8] (That is why, in this case $S$ is in the stopping region and If we suppose that $S>E$, then $(E-S)^+=0$, that is exercise payoff is zero and thus $S$ belongs to the continuation region. It contradicts $V(S,t)=(E-S)^+$.) Thus if $V(S,t)=(E-S)^+$ then we can rewrite as $V(S,t)=E-S$. And $\eqref{eq 24}$ is written as
$$
\min\{-LV,V-(E-S)\}=0.
$$
Since $(V-(E-S))=0$ , then we have $-LV=rE-qS\ge0$ and if $q\ge0$, then we have $S\le{rE/q}$.(QED)
\end{proof}

\begin{lemma}
Let $V(S,t)$ be the solution to $\eqref{eq 24}$ and fix $t>0$. Then\\
1)	If there exists such $S_0>0$ that $V(S_0,t)=(E-S_0)^+$ , then for $\forall{S<S_0}$, we have $V(S,t)=(E-S)^+$.\\
2)	If there exists such $S_1>0$ that $V(S_1,t)>(E-S_1)^+$, then for $\forall{S>S_1}$, we have $V(S,t)>(E-S)^+$.
\end{lemma}

\begin{proof}
1) From Lemma 1, we have $S_0\le{rE/q}$. Suppose that the conclusion were not true, that is, suppose that
\begin{equation}\label{eq *}
0<\exists{S<S_0}:V(S,t)>E-S.\tag{*}
\end{equation} 
Let $(a,b)(b\le{S_0})$ be the longest interval where holds $\eqref{eq *}$. Then, we have $V(S,t)=(E-S)$ at $S=a$ and $S=b$. From $\min\{-LV,V-\phi\}=0$, we have $-LV=0,V-(E-S)>0$, on that interval $(a, b)$. Thus we have
$$
-L(V-(E-S))=L(E-S)=qS-rE<0~~(\because{q\ge0}\Rightarrow{S<S_0\le{rE/q}}).
$$
Thus from the conclusion of 2) of Theorem 1, $V-(E-S)$ attains the nonnegative maximum value at the boundary. But $V(a,t)-(E-S)=0,V(b,t)-(E-S)=0$, so we have $V(S,t)-(E-S)\le0,~~a<\forall{S}<b$. This contradicts $\eqref{eq *}$.\\
2) If $\exists{S_2>S_1}:V(S_2,t)=(E-S_2)^+$, then from the conclusion of 1) we have $V(S_1,t)=(E-S_1)^+$. This contradicts the assumption (QED).
\end{proof}

\begin{theorem}
(Existence of exercise boundary) When $V(S,t)$ is the solution to $\eqref{eq 24}$ and decreasing on $t$, then for any $t>0$, there exists such $s(t)(\le\min\{rE/q,E\})$ that if $0<S<s(t)$ then $V(S,t)=E-S$ and if $s(t)<S$ then $V(S,t)>(E-S)^+$.
\end{theorem}

\begin{proof}
Fix $t>0$, then $(E,+\infty)\times\{t\}$ is included in the stopping region, so the stopping region is not empty. Let $s(t)$ be the infimum of the stopping region, then from Lemma 2, $(s(t),+\infty)$ becomes the stopping region and $(0,s(t))$ becomes the continuation region (QED).
\end{proof}

\begin{theorem}
(Uniqueness of the solution) The solution to $\eqref{eq 24}$ that is decreasing on $t$ is unique.
\end{theorem}

\begin{proof}
Let $V_1,V_2$ be the two solutions to $\eqref{eq 24}$ which are decreasing on $t$ and   be the exercise boundaries of two solutions, respectively. Fix $t>0$. Without loss of generality, assume that $s_1(t)>s_2(t)$. Then if $S<s_2(t)$, then $V_1=V_2=E-S$, and if $S>s_1(t)$ then $LV_1=LV_2=0$. If $s_2(t)<S\le{s_1(t)}$ then we have
\begin{equation}\label{eq **}
V_1=E-S,~~LV_2=0,~~V_2>E-S.
\end{equation}
Now consider $V_2-V_1$ in the interval $(s_2(t),\infty)$. When $s_2(t)<S\le{s_1(t)}$, we have
$$
-L(V_2-V_1)=L(E-S)=qS-rE<0~~(s<s_1(t)\le{rE/q}).
$$ 
and when $S>s_1(t)$, we have $-L(V_2-V_1)=0$, so finally we have $-L(V_2-V_1)\le0$ in the interval $(s_2(t),\infty)$.
From 2) of Theorem 1, the nonnegative maximum value of $V_2-V_1$ is attained at the boundary in the interval $[s_2(t),\infty)$. But we have $(V_2-V_1)(s_2(t),t)=0$ and $(V_2-V_1)(\infty,t)=0$  in $[s_2(t),\infty)$. Thus we have $V_2-V_1\le0$ in the interval $[s_2(t),\infty)$. This contradicts $\eqref{eq **}$. So we have $s_1(t)=s_2(t),\forall{t>0}$.\\
Now let $s(t)$ be the exercise boundary. Two solutions are equal to $E-S$ on the interval $(0,s(t))$. In the interval $(s_1(t),\infty)$, we have $-L(V_1-V_2)=0$ and $V_1-V_2|_{s(t),+\infty}=0$, so we have $V_1=V_2$ from Theorem 1(QED).
\end{proof}

\subsection{The limit of the BTM price of American Option.}

\indent

Suppose that $T>0$ is the maturity. Let $N$ be the number of partition intervals, $\Delta{t}=T/N$ the length of the partition interval and $t_n=n\Delta{t},~~n=0,1,\dots,N.$  Especially $T=t_N$.\\
Let
\begin{equation}\label{eq 26}
u=e^{\sigma\sqrt{\Delta{t}}},~\theta=\frac{\rho/\eta-d}{u-d},~\rho=1+r\Delta{t},~\eta=1+q\Delta{t},~ud=1
\end{equation}

Here $r, q$ and $\sigma$ are interest rate, dividend rate and volatility, respectively. In BTM we let  $S_j=S_0u^j,~j\in{Z}$. Denote by $V_j^n$ the American option price at time $t_n$ with underlying asset value $S_j$. Then American option price by BTM is as follows \cite{JD, Jia}:
\begin{eqnarray}\label{eq 27}
&&V_j^N=\varphi,~~~j\in{Z}\nonumber\\
&&V_j^{k-1}=\max\left\{\frac{1}{\rho}\left(\theta{V_{j+1}^k}+(1-\theta){V_{j-1}^k}\right),\phi_j\right\},~~~k=N,N-1,\dots,1
\end{eqnarray}
Here $\phi_j=(S_j-E)^+$ (call) or $\phi_j=(E-S_j)^+$(put).
\begin{theorem}
Let $V_j^n(T)~(j\in{Z})$ be the price by BTM of American option with the maturity $T$. Then $V_j^n(T)$ is increasing on $T$. That is, if $T_1=t_N<T_2=t_M$, then
\begin{equation}\label{eq 28}
V_j^n(T_1)\le{V_j^n(T_2)}~~~n=0,1,\dots,N.
\end{equation}
\end{theorem}

\begin{proof}
1) First, assume that $M=N+1$. From $\eqref{eq 27}$, $V_j^N(T_1)=\phi_j,V_j^{N+1}(T_2)=\phi_j$, and we have
\begin{equation*}
V_j^N(T_2)=\max\left\{\frac{1}{\rho}\left(\theta\phi_{j+1}+(1-\theta)\phi_{j-1}\right),\phi_{j}\right\}.
\end{equation*}
So we have $V_j^N(T_2)\ge\phi_j=V_j^N(T_1)$. That is, when $n=N$, we have $\eqref{eq 28}$.
Assume that we have $\eqref{eq 28}$ when $n=k$, then we will prove $\eqref{eq 28}$ when $n=k-1$.
That is, we will prove that if $V_j^k(T_1)\le{V_j^k(T_2)}$ then $V_j^{k-1}(T_1)\le{V_j^{k-1}(T_2)}$.In fact
\begin{eqnarray*}
V_j^{k-1}(T_1)&=&\max\left\{\frac{1}{\rho}\left(\theta{V_{j+1}^k(T_1)}+(1-\theta)V_{j-1}^k(T_1)\right),\phi_j\right\}\\
&\le&\max\left\{\frac{1}{\rho}\left(\theta{V_{j+1}^k(T_2)}+(1-\theta)V_{j-1}^k(T_2)\right),\phi_j\right\}=V_j^{k-1}(T_2).
\end{eqnarray*} 
2) Now let us prove when $0=t_0<t_1<\cdots{t_N}=T_1<t_{N+1}<\cdots<t_{N+r}=T_2$ in general.
From the result of 1), when $n=0,1,\dots,N$, we have
$$V_j^n(T_1)=V_j^n(t_N)\le{V_j^n(t_{N+1})}\le{V_j^n({t_{N+2}})}\le\cdots\le{V_j^n{(t_{N+r}})}=V_j^n(T_2),\forall{j\in{Z}}.$$
So the required result is proved. (QED)
\end{proof}
{\\}

{\bf Remark 1.} 
The price by BTM of American call option $\phi_j=(S_j-E)^+$ satisfies $\phi_j\le{V_j^n(T)}\le{S}$ and is increasing on $S$, and decreasing on $t$ ($V_j^n\le{V_{j+1}^n}$ and $V_j^n\le{V_j^{n-1}}$). The price by BTM of American put option $(\phi_j=(E-S_j)^+)$ satisfies $\phi_j\le{V_j^n(T)\le{E}}$ and is decreasing on $S$ and $t$ ($V_{j+1}^n\le{V_j^n}$ and $V_j^n\le{V_j^{n-1}}$). And $V_{-\infty}^n(T)=E,V_{+\infty}^n(T)=0$ [8, 17]

So the price by BTM of American option converges when $T\rightarrow\infty$. We denote this limit by $U_j^n(j\in{Z},n=0,1,\dots)$. That is,
\begin{equation}\label{eq 29}
U_j^n=\lim_{T\rightarrow\infty}V_j^n(T),~~\forall{j\in{Z}},~n=0,1,\dots.
\end{equation}

\begin{lemma}
The limit of the BTM price of call option is increasing on $S(j)$, and is  decreasing on $t(n)$ and the limit of the BTM price of put option is decreasing on $S$ and $t$. These limit prices have the same boundedness as Remark 1.
\end{lemma}

{\bf Time independence of the limit $U_j^n$ of the put option BTM price.}

From (2), the limit $U_j^n$ of the BTM price of put option satisfies the following difference equation.
\begin{eqnarray}\label{eq 30}
&&U_j^k=\max\left\{\frac{1}{\rho}\left(\theta{U_{j+1}^{k+1}}+(1-\theta)U_{j-1}^{k+1}\right),\phi_j\right\},~~j\in{Z},~k=0,1,\dots\nonumber\\
&&U_{-\infty}^k=E,~~U_\infty^k=0.
\end{eqnarray}

Denote by $l_\infty(Z)$ the Banah space formed with two-sided sequence of real numbers. For ${\bf U}\in{l_\infty(Z)}$, the norm $\Vert{\bf U}\Vert$ is defined as follows:  
$$
\Vert{\bf U}\Vert:=\sup_{j\in{Z}}|U_j|.
$$
Define the operator $B:l_\infty(Z)\rightarrow{l_\infty(Z)}$ as follows:
\begin{equation}\label{eq 31}
(B{\bf U})_j=\max\left\{\frac{1}{\rho}\left(\theta{U_{j+1}}+(1-\theta)U_{j-1}\right),\phi_j\right\},~~j\in{Z},~{\bf U}\in{l_\infty(Z)}.
\end{equation}
Then
\begin{eqnarray}\label{eq 32}
\Vert{B{\bf U}-B{\bf V}}\Vert&=&sup_j\left\{\max\left[\frac{1}{\rho}\left(\theta{U_{j+1}}+(1-\theta)U_{j-1}\right),\phi_j\right]-\max\left[\frac{1}{\rho}\left(\theta{V_{j+1}}+(1-\theta)V_{j-1}\right),\phi_j\right]\right\}\nonumber\\
&\le&sup_j\left\{\max\left[\frac{1}{\rho}\left(\theta({U_{j+1}}-V_{j+1})+(1-\theta)(U_{j-1}-V_{j-1})\right),0\right]\right\}\\
&\le&\max\left[\frac{1}{\rho}\left(\theta{sup_j}|U_{j+1}-V_{j+1}|+(1-\theta)sup_j|U_{j-1}-V_{j-1}|\right),0\right]\nonumber\\
&\le&\frac{1}{\rho}\Vert{{\bf U}-{\bf V}}\Vert.\nonumber
\end{eqnarray}
Thus we proved the following theorem.

\begin{lemma}
Operator $B:l_\infty(Z)\rightarrow{l_\infty(Z)}$ is contraction mapping in Banach space and has a unique fixed point when $r>0$.
\end{lemma}
Now consider the time independence of $U_j^n$ defined by $\eqref{eq 29}$. Denote by
$$
{\bf U}^n=\left\{U_j^n\right\}_{j\in{Z}}\in{l_\infty(Z)}.
$$
From $\eqref{eq 30}$, ${\bf U}^k=B{\bf U}^{k+1},~~k=0,1,\cdots$ and thus using ($\eqref{eq 32}$) repeatedly then we have 
$$
\Vert{{\bf U}^0-{\bf U}^1}\Vert=\Vert{B{\bf U}^1-B{\bf U}^2}\Vert\le\frac{1}{\rho}\Vert{{\bf U}^1-{\bf U}^2}\Vert\le\frac{1}{\rho^2}\Vert{{\bf U}^2-{\bf U}^3}\Vert\le\dots\le\frac{1}{\rho^n}\Vert{{\bf U}^n-{\bf U}^{n+1}}\Vert.
$$

So we have $\Vert{{\bf U}^0-{\bf U}^1}\Vert\le{2E/\rho^n}$ because $\Vert{{\bf U}^n-{\bf U}^{n+1}}\Vert\le{2E}$. thus we have ${{\bf U}^0-{\bf U}^1}$, and we have ${{\bf U}^k-{\bf U}^{k+1}},~k=0,1,\cdots$ in the same way. Thus we proved the following theorem.

\begin{theorem}
The limit $U_j^n$ of the BTM price of American put option is independent on $n$. Thus $U_j=U_j^n$ is the solution to the following problem.
\begin{eqnarray}\label{eq 33}
&&\max\left\{\frac{1}{\rho}\left(\theta{U_{j+1}}+(1-\theta)U_{j-1}\right),\phi_j\right\},~~j\in{Z}.\nonumber\\
&&U_{-\infty}=E,~~~U_\infty=0
\end{eqnarray}
\end{theorem}

{\bf Remark 2.} So the limit of the BTM price of American put option is the BTM price of perpetual American option \cite{LL}.

\subsection{The Limit of the price by Explicit Difference Scheme of variational inequality of American option.}

\indent

Let $V(S,t,T)$ be the price of American option with expiry date $T$, that is, the solution to the  variational inequality
\begin{eqnarray}\label{eq 34}
&&\min\left\{-\frac{\partial{V}}{\partial{t}}-\frac{\sigma^2}{2}S^2\frac{\partial^2V}{\partial{S^2}}-(r-q)S\frac{\partial{V}}{\partial{S}}+rV,V-\phi\right\}=0,~~0<t<T,~S>0,\nonumber\\
&&V(T,S)=\phi(S)=(E-S)^+~~or~~(S-E)^+.
\end{eqnarray}

Then $V(S,t,T)$ is increasing on $T$ \cite{Jia}. When $\Delta{t}=T/N,\Delta{x}>0$, let
$$
S_j=S_0{e^{j\Delta{x}}}=e^{j\Delta{x}+c},~~j\in{Z},~t_n=n\Delta{t}.
$$
Denote
$$
U_j^n=V(S_j,t_n),~w=\frac{\sigma^2\Delta{t}}{\Delta{x}^2},~a=\frac{1}{2}+(r-q-\frac{\sigma^2}{2})\frac{\Delta{x}}{2\sigma^2},~\rho=1+r\Delta{t}.
$$
Then the explicit difference scheme for $\eqref{eq 34}$ is privided as follows \cite{JD, Jia}:
\begin{eqnarray}\label{eq 35}
&&U_j^N=\phi_j,j\in{Z}\\
&&U_j^n=\max\left\{\frac{1}{\rho}\left\{(1-w)U_j^{n+1}+w\left[aU_{j+1}^{n+1}+(1-a)U_{j-1}^{n+1}\right]\right\},\phi_j\right\},~n=N-1,\cdots,1,0\nonumber
\end{eqnarray}

\begin{theorem}
Let $U_j^n(T)$ be the solution to the explicit difference scheme of American Option with the maturity $T$. Then $U_j^n(T)$ is increasing on $T$. That is, if $T_1=t_N<T_2=t_M$, then we have $U_j^n(T_1)\le{U_j^n(T_2)}$.
\end{theorem}

\begin{proof}
1) First, let $M=N+1$, that is,
$$
0=t_0<t_1<\cdots{t_N}=T_1<t_{N+1}<T_2=t_{N+2}.
$$ 
Then $U_j^N(T_1)=\phi_j,U_j^{N+1}(T_2)=\phi_j$ and we have
\begin{eqnarray*}
&&U_j^N(T_1)=\max\left\{\frac{1}{\rho}\left\{(1-w)U_j^{N+1}+w\left[aU_{j+1}^{N+1}+(1-a)U_{j-1}^{N+1}\right]\right\},\varphi_j\right\}\\
&&\ge\varphi=U_j^N(T_1).
\end{eqnarray*}
So when $n=N$, we have the result of the theorem. Assume that $U_j^n(T_1)\le{U_j^n(T_2)}$ when $n=k$, and we will prove when $n=k-1$.
\begin{eqnarray*}
U_j^{k-1}(T_1)&=&\max\left\{\frac{1}{\rho}\left\{(1-w)U_j^k(T_1)+w\left[aU_{j+1}^k(T_1)+(1-a)U_{j-1}^k(T_1)\right]\right\},\varphi_j\right\}\\
&\le&\max\left\{\frac{1}{\rho}\left\{(1-w)U_j^k(T_2)+w\left[aU_{j+1}^k(T_2)+(1-a)U_{j-1}^k(T_2)\right]\right\},\varphi_j\right\}\\
&=&U_j^k(T_2)~\forall{j\in{Z}}.
\end{eqnarray*}
2) In general, consider the case when
$$
0=t_0<t_1<\cdots<t_N=T_1<t_{N+1}<\cdots<t_{N+r}=T_2,~~r\in{N}.
$$
From the result of 1), we have 
$$
U_j^n(T_1)=U_j^n(t_N)\le{U_j^n(t_{N+1})}\le{U_j^n(t_{N+2})}\le\cdots\le{U_j^n(t_{N+r})}=U_j^n(T_2),~~\forall{j\in{Z}}.
$$
Thus, we have proves the required result.(QED)
\end{proof}
{\\}

{\bf Remark 3.}
The price by explicit difference scheme of American call option $(\phi_j=(S_j-E)^+)$ satisfies $\phi_j\le{U_j^n(T)\le{S_j}}$ and is increasing on $S$ and decreasing on $t$ $( U_j^n\le{U_{j+1}^n}~and~U_j^n\le{U_j^{n-1}})$. The price by explicit difference scheme of American put option $\phi_j=(E-S)^+$ satisfies $\phi_j\le{U_j^n(T)\le{E}}$ and is decreasing on $S$ and $t$ $(U_{j+1}^n\le{U_j^n},U_j^n\le{U_j^{n-1}})$, and $U_{-\infty}^n(T)=E,~U_{+\infty}^n(T)=0$ \cite{JD, Jia}.
{\\}

Thus the prices by explicit difference scheme of American options with the maturity $T$ converge when $T\rightarrow\infty$. Denote the limit by $W_j^n$($j\in{Z},~n=0,1,\dots$ ). That is,
\begin{equation}\label{eq 36}
W_j^n=\lim_{T\rightarrow\infty}U_j^n(T),~~\forall{j\in{Z}},~n=0,1,\cdots
\end{equation}
{\\}
The limit of the price by explicit difference scheme of American call option is increasing on $S(j)$ and is decreasing on $t(n)$ and the limit of put option is decreasing on $S$ and $t$. They have the same boundedness as pointed in Remark 3.
{\\}

From $\eqref{eq 35}$, the limit $W_j^n$ of the price by explicit difference scheme of American put option $(\phi_j=(E-S_j)^+)$ satisfies the following difference equation.
\begin{eqnarray}\label{eq 37}
&&W_j^n=\max\left\{\frac{1}{\rho}\left\{(1-w)W_j^{n+1}+w\left[aW_{j+1}^{n+1}+(1-a)W_{j-1}^{n+1}\right]\right\},\phi_j\right\},~~j\in{Z},~n=0,1,\dots.\nonumber\\
&&W_{-\infty}^k=E,~~W_{\infty}^k=0.
\end{eqnarray}
This is just the difference equation of \eqref{eq 24}, the variational inequality on infinite time interval. 
Define an operator $F:l_\infty(Z)\rightarrow{l_\infty(Z)}$ as follows:
\begin{equation}\label{eq 38}
(F{\bf U})_j=\max\left\{\frac{1}{\rho}\left\{(1-w)U_j+w\left[aU_{j+1}+(1-a)U_{j-1}\right]\right\},\phi_j\right\},~~j\in{Z},~{\bf U}\in{l_\infty{Z}}.
\end{equation}
Then we have
\begin{eqnarray}\label{eq 39}
&&\Vert{F{\bf U}-F{\bf V}}\Vert=sup_j\left\{\max\left\{\frac{1}{\rho}\left\{(1-w)U_j+w\left[aU_{j+1}+(1-a)U_{j-1}\right]\right\},\phi_j\right\}-\right.\nonumber\\
&&~~~~~~~~~~~~~~~~~-\max\left.\left\{\frac{1}{\rho}\left\{(1-w)V_j+w\left[aV_{j+1}+(1-a)V_{j-1}\right]\right\},\phi_j\right\}\right\}\nonumber\\
&&\le{sup_j}\left\{\max\left[\frac{1}{\rho}\left\{(1-w)(U_j-V_j)+w\left[a(U_{j+1}-V_{j+1})+(1-a)(U_{j-1}-V_{j-1})\right]\right\},0\right]\right\}\\
&&\le\max\left[\frac{1}{\rho}\left((1-w)sup_j|U_j-V_j|+w\left[asup_j|U_{j+1}-V_{j+1}|+(1-a)sup_j|U_{j-1}-V_{j-1}|\right]\right),0\right]\nonumber\\
&&\le{\frac{1}{\rho}\Vert{{\bf U}-{\bf V}}\Vert}\nonumber
\end{eqnarray}
So we proved the following lemma.
\begin{lemma}
The operator $F:l_\infty(Z)\rightarrow{l_{\infty}(Z)}$ is contraction mapping and has a unique fixed point when $r>0$.
\end{lemma}
Now consider the time independence of $W_j^n$ defined by $\eqref{eq 37}$. Let 
${\bf W}^n=\{W_j^n\}_{j\in{Z}}\in{l_\infty(Z)}.$ Then from $\eqref{eq 30}$, we have ${\bf W}^n=F{\bf W}^{n+1},~n=0,1,\dots$ and thus using $\eqref{eq 39}$ repeatedly, we have 
$$
\Vert{{\bf W}^0-{\bf W}^1}\Vert=\Vert{F{\bf W}^1-F{\bf W}^2}\Vert\le\frac{1}{\rho}\Vert{{\bf W}^1-{\bf W}^2}\Vert\le\frac{1}{\rho^2}\Vert{{\bf W}^2-{\bf W}^3}\Vert\le\cdots\le\frac{1}{\rho^n}\Vert{{\bf W}^n-{\bf W}^{n+1}}\Vert.
$$
So we have $\Vert{{\bf W}^0-{\bf W}^1}\Vert\le{2E/\rho^n}\rightarrow0,~n\rightarrow\infty$ because $\Vert{{\bf W}^n-{\bf W}^{n+1}}\Vert\le{2E}$. Thus we have ${\bf W}^0={\bf W}^1$ and in the same way we have ${{\bf W}^n-{\bf W}^{n+1}},~~n=0,1,\dots$. Thus we proved the following theorem.
\begin{theorem}
The limit $W_j^n$ of the price by explicit difference scheme of American put option is independent on $n$. That is, $W_j=W_j^n$ is the solution to the following problem.
\begin{eqnarray}\label{eq 40}
&&W_j=\max\left\{\frac{1}{\rho}\left\{(1-w)W_j+w\left[aW_{j+1}+(1-a)W_{j-1}\right]\right\},\phi_j\right\},~~j\in{Z}\nonumber\\
&&W_{-\infty}=E,~W_{\infty}=0.
\end{eqnarray}
\end{theorem}

{\bf Remark 2.} So the limit of the price by explicit difference scheme of American put option becomes the price by the difference equation \eqref{eq 7} of perpetual American put option.



\end{document}